\documentclass[11pt]{article}
\usepackage[T1]{fontenc}
\usepackage{amsmath,amsxtra,amssymb,latexsym,amscd,amsthm,pb-diagram,fancyhdr,euscript}
\usepackage{amsthm}
\usepackage{indentfirst}
\usepackage{epsfig}
\usepackage{mathrsfs}
\usepackage{slashed}
\usepackage{hyperref}
\hypersetup{
   colorlinks,
    citecolor=red,
    filecolor=black,
    linkcolor=blue,
    urlcolor=magenta
}
\hypersetup{linktocpage}

\renewcommand{\d}{\mathrm{d}}

\newcommand{\scri}{{\mathscr I}}

\newcommand{\hook}{{\setlength{\unitlength}{11pt}   
                   \begin{picture}(.833,.8)
                   \put(.15,.08){\line(1,0){.35}}
                   \put(.5,.08){\line(0,1){.5}}
                   \end{picture}}}
\newtheorem{definition}{Definition}
\newtheorem{theorem}{Theorem}
\newtheorem{proposition}{Proposition}

\newtheorem{lemma}{Lemma}
\newtheorem{remark}{Remark}

\begin{document}
\mbox{} \thispagestyle{empty}

\begin{center}
\bf{\Huge Conformal scattering theory for the linearized gravity fields on Schwarzschild spacetime} \\

\vspace{0.1in}

{PHAM Truong Xuan\footnote{Departement of informatics, faculty of mathematics, Thuyloi University, 175 Tay Son, Dong Da, Ha Noi, Viet Nam. Email~: xuanpt@tlu.edu.vn or phamtruongxuan.k5@gmail.com}}
\end{center}

{\bf Abstract.} We provide in this paper a first step to obtain the conformal scattering theory for the linearized gravity fields on the Schwarzschild spacetime by using the conformal geometric approach. We will show that the existing decay results for the solutions of the Regge-Wheeler and Zerilli equations obtained recently by L. Anderson, P. Blue and J. Wang \cite{ABlu} is sufficient to obtain the conformal scattering.



\tableofcontents

\section{Introduction}
Conformal scattering is a geometric approach to the scattering theory. The idea was started by Penrose \cite{Pe1964} by constructing the conformal compactifications of the spacetimes. However, the formulation of theory is established really by Friedlander. First, he gave the notion of radiation fields \cite{Fri1962,Fri1964,Fri1967}, then he constructed initially the conformal scattering theory \cite{Fri1980} to the linear wave equations on a static spacetime where the metric approaches the flat metric fast enough. The method is based on solving the characteristic Cauchy problem i.e Goursat problem on the null infinity hypersurfaces $\scri^\pm$, where the data are the radiation fields. Later, the conformal scattering theory has been constructed and developed by Baez, Segal and Zhou \cite{BaSeZho1990}, H\"ormander \cite{Ho1990} for the linear and nonlinear wave equations on the Minkowski spacetime and again by Friedlander \cite{Fri2001} for the linear wave equations asymptotically Euclidean manifolds.

In 2004, Mason and Nicolas constructed the conformal scattering for linear wave, Dirac and Maxwell equations on the asymptotical simple spacetimes i.e non static spacetimes \cite{MaNi2004}. Then Joudioux worked on the same geometrical background to produce the conformal scattering for a nonlinear wave equation \cite{Jo2012,Jo2019}.
After that, Nicolas and Mokdad constructed the conformal scattering on the detailed static spacetimes such as Schwarzschild for linear wave equations \cite{Ni2016} and on Reissner-Nordstr\"om-de Sitter for Maxwell equations \cite{Mo2019}. In their work, Nicolas and Mokdad use the energy and pointwise decay results of the fields to obtain the energy identity up to the null infinity $i^+$. Then, they solve the Goursat problem and show that the trace operator is an isomorphism. Recently, Taujanskas produces a conformal scattering of the Maxwell-scalar field system on de Sitter space \cite{Ta2019}.

Moreover, there are some related works of Dafermos et al. \cite{Da2018}, Kehle et al. \cite{Ke2019} and Angelopoulos et al. \cite{An2020} which also used the energy and pointwise decay results to construct scattering theories for the scalar wave equations on the slow Kerr spacetime, on the interior of Reissner-Nordstr\"om black hole and on extremal Reissner-Nordstr\"om respectively.

A conformal scattering theory on the asymptotic flat spacetimes consists on three steps: first, the resolution of the Cauchy problem of the rescaled equations on the rescaled spacetime, then the definition and the extension of the trace operator. Second, we prove the energy identity up to the future timelike infinity $i^+$, this shows that the extension of the trace operator is injective. Third, we solve the Goursat problem with the initial data on the conformal boundary consisting the horizon and null infinity. Thereofre, the trace operator will be subjective, hence, an isometry.  

Conformal scattering theory for the linearized gravity fields on the Minkowski spacetime is mentioned in the thesis of the author (see \cite{Pha2017}) by studying the $n/2$-spin zero rest-mass fields which are general cases of linearized gravity fields in the spin framework. The Goursat problem has been solved in the Penrose conformal compactification of the spacetime. Howerver, the conformal scattering operator in this case is evident, so the construction of the theory is not significant in the flat spacetime. Considering the problem further, in this paper, we develop the works of Nicolas and Mokdad \cite{Ni2016,Mo2019} to construct a conformal scattering for the scalar Regge-Wheeler and Zerilli equations. These are the equations governing the perturbations of the vacuum Schwarzschild metric (see \cite{ReWhe,Vis,Ze}). They are both scalar wave type with the real potentials which decay as $r^{-3}$. This leads to the conclusion that the conformal rescaled equations of the scalar Regge-Wheeler and Zerilli equations are regular and very like the rescaled wave equations.  Hence, it can be seem to construct the conformal scattering theories for these equations as well as the linear scalar wave equation (see the construction of the scalar wave equation in \cite{Ni2016}). The main obstacles are: what are the energy and pointwise decays of the solutions that are necessary to obtain the conformal scattering theory? And how we can show that the energy and pointwise decays are sufficient to obtain the energy identity up to the future timelike infinity $i^+$?

We will use the Penrose's conformal compactification of the exterior domain of the Black hole to our construction. This domain is foliated by the hypersurfaces $\left\{\mathcal{H}_{\tau}\right\}_{\tau}$. 
The hypersurface ${\mathcal{H}}_{\tau}$ tends to the null infinity $i^+$ and consists on three portions: one is parallel to the horizon $\mathfrak{H}^+$, one is parallel to $\Sigma_0$ and the last is parallel to $\scri^+$.
We will work on a global time-like vector field $K$ which transverses the horizon and equals to $\partial_t$ as the radial parameter $r$ sufficient large (see Section 2.2 and more details in \cite{ABlu}). This time-like vector field plays an important role in the study of the decay results of the linearized gravity fields. We also use the time-like Killing vector field $\partial_t$ to obtain the energy of the fields on the Cauchy hypersurface $\Sigma_T=\left\{ t=T \right\}$ in the construction of the theory. The Cauchy problem for the equations and the rescaled ones are obtained by Leray's theorem and the energy identity such as the previous work for non-linear Klein Gordon equations in \cite{Ni1995}. Therefore, we can define the trace operator, then it can be extended by density (see Section 3.2). 

Using the energy and pointwise decays obtained in the recent work of Anderson et al. \cite{ABlu} we will prove that the energy across the hypersurface $\mathcal{H}_T$ tends to zero when $T$ tends to infinity. This leads to the energy indentity up to $i^+$ i.e the energy of the fields restricted on the initial hypersurface $\Sigma_0=\left\{ t=0 \right\}$ equals to the sum of energies of the fields restricted on the future horizon $\mathfrak{H}^+$ (resp. $\mathfrak{H}^-$) and on the future null infinity $\scri^+$ (resp. $\scri^-$). 
As a consequence, the extended trace operator is injective (see Section 3.3).

To prove that the trace operator is subjective, we solve the Goursat problem with the smooth compactly supported initial data on the conformal boundary. This work is done by applying the results of H\"ormander on Schwarzschild spacetime, the method is
similar to the scalar wave equation in \cite{Ni2016} (see Section 4).\\ 
{\bf Notation.}\\ 
We denote the space of smooth compactly supported scalar functions on $\mathcal{M}$ (a smooth manifold without boundary) by $\mathcal{C}_0^\infty(\mathcal{M})$ and the space of distributions on $\mathcal{M}$ by $\mathcal{D}'(\mathcal{M})$.

\section{Geometrical and analytical setting} 

\subsection{Schwarzschild metric and Penrose's conformal compactification}
Let $({\cal{M}}=\mathbb{R}_t\times ]0,+\infty[_r\times S^2_\omega,g)$ be a $4$-dimensional Schwarzschild manifold, equipped with the Lorentzian metric $g$ given by
$$g = F \d t^2 - F^{-1}\d r^2 - r^2\d\omega^2, \, F=F(r)=1-\frac{2M}{r},$$
where $\d \omega^2$ is the Euclidean metric on $S^2$, and $M>0$ is the mass of the black hole. In this paper, we work on the exterior of the black hole $\mathcal{B}_I:=\left\{ r>2M \right\}$.

We denote the Regge-Wheeler variable by $r_*=r+2M \log(r-2M)$. We have that $\d r=F\d r_*$. In the coordinates $(t,r_*,\omega^2)$, the Schwarzschild metric takes the form
$$g = F(\d t^2- \d r_*^2) - r^2\d\omega^2.$$
The retarded and advanced Eddington-Finkelstein coordinates $u$ and $v$ are defined by
$$u=t-r_*, \, v= t+r_*.$$
Then, In the coordinates $(u,v,\omega^2)$ the Schwarzschild metric takes the form
$$g = F\d u \d v - r^2\d\omega^2.$$
The outgoing and incoming principal null directions are respectively,
$$l =\partial_v = \partial_t + \partial_{r_*} \hbox{  and  } n = \partial_u = \partial_t - \partial_{r_*}.$$

By taking $\Omega = 1/r$ and $\hat{g} = \Omega^2g$, we yield a compactification of the exterior domain in the coordinates $u, \, R = 1/r, \, \omega$ that is $\left(\mathbb{R}_u\times \left[0,\dfrac{1}{2M}\right] \times S^2_\omega, \hat{g} \right)$ with the rescaled metric 
\begin{equation}
\hat{g} = R^2(1-2MR) \d u^2 - 2\d u\d R - \d\omega^2.
\end{equation}
Future null infinity $\scri^+$ and the past horizon $\mathfrak{H}^-$ are null hypersurfaces of the rescaled spacetime
$$\scri^+ = \mathbb{R}_u \times \left\{ 0\right\}_R \times S_\omega^2, \, \mathfrak{H}^- = \mathbb{R}_u \times \left\{ 1/2M\right\}_R \times S^2_\omega.$$
By using the advanded coordinates $v, \, R, \, \omega$, the rescaled metric $\hat{g}$ takes the form
\begin{equation}
\hat{g} = R^2(1-2MR)\d v^2 + 2 \d v \d R - \d\omega^2.
\end{equation}
Past null infinity $\scri^-$ and the future horizon $\mathfrak{H}^+$ described as the null hypersurfaces
$$\scri^- = \mathbb{R}_v \times \left\{ 0\right\}_R \times S_\omega^2, \, \mathfrak{H}^+ = \mathbb{R}_v \times \left\{ 1/2M\right\}_R \times S^2_\omega.$$
The Penrose conformal compactification of $\mathcal{B}_I$ is 
$$\bar{\mathcal{B}_I} = \mathcal{B}_I \cup \scri^+ \cup \mathfrak{H}^+ \cup \scri^-\cup \mathfrak{H}^-\cup S_c^2,$$
where $S_c^2$ is the crossing sphere.

The volume forms associated to the Schwarzschild metric $g$ and the rescaled metric $\hat{g}$ are given by
$$\mathrm{dVol}_g = r^2F \d t \wedge \d r_* \wedge \d^2\omega = \frac{r^2F}{2}\d u \wedge\d v \wedge \d^2\omega$$
and 
$$\mathrm{dVol}_{\hat g} = R^2F \d t \wedge \d r_* \wedge \d^2\omega = \frac{R^2F}{2} \d u \wedge \d v \wedge \d^2\omega,$$
respectively, where $\d^2\omega$ is the euclidean area element on the unit $2$-sphere $S^2$.

\subsection{Foliation and timelike vector fields}
Following \cite{ABlu}, we take $r_{NH}$ and $r_{FH}$ such that $2M<r_{NH}<3M < r_{FH}$, where $r_{FH}$ is sufficiently large. The future $\mathcal{J}^+(\Sigma_0)\, (\hbox{where } \Sigma_0=\left\{ t=0 \right\})$ is foliated by the hypersurfaces $\mathcal{H}_\tau:=\Sigma_\tau^i \cup \Sigma_\tau^e \cup \mathcal{N}_\tau\, \,  (\tau \geq 0)$, where,
$$\Sigma_\tau^i = \left\{ v=\tau + (r_{NH})_* \right\} \cap \left\{ r<r_{NH}\right\},$$
$$\Sigma_\tau^e = \left\{ t=\tau\right\} \cap \left\{ r_{NH} \leq r\leq r_{FH} \right\}$$
and
$$\mathcal{N}_\tau = \left\{ u=\tau - (r_{FH})_* \right\} \cap \left\{ v\geq \tau + (r_{FH})_* \right\}.$$
\begin{figure}
\begin{center}
\includegraphics[scale=0.5]{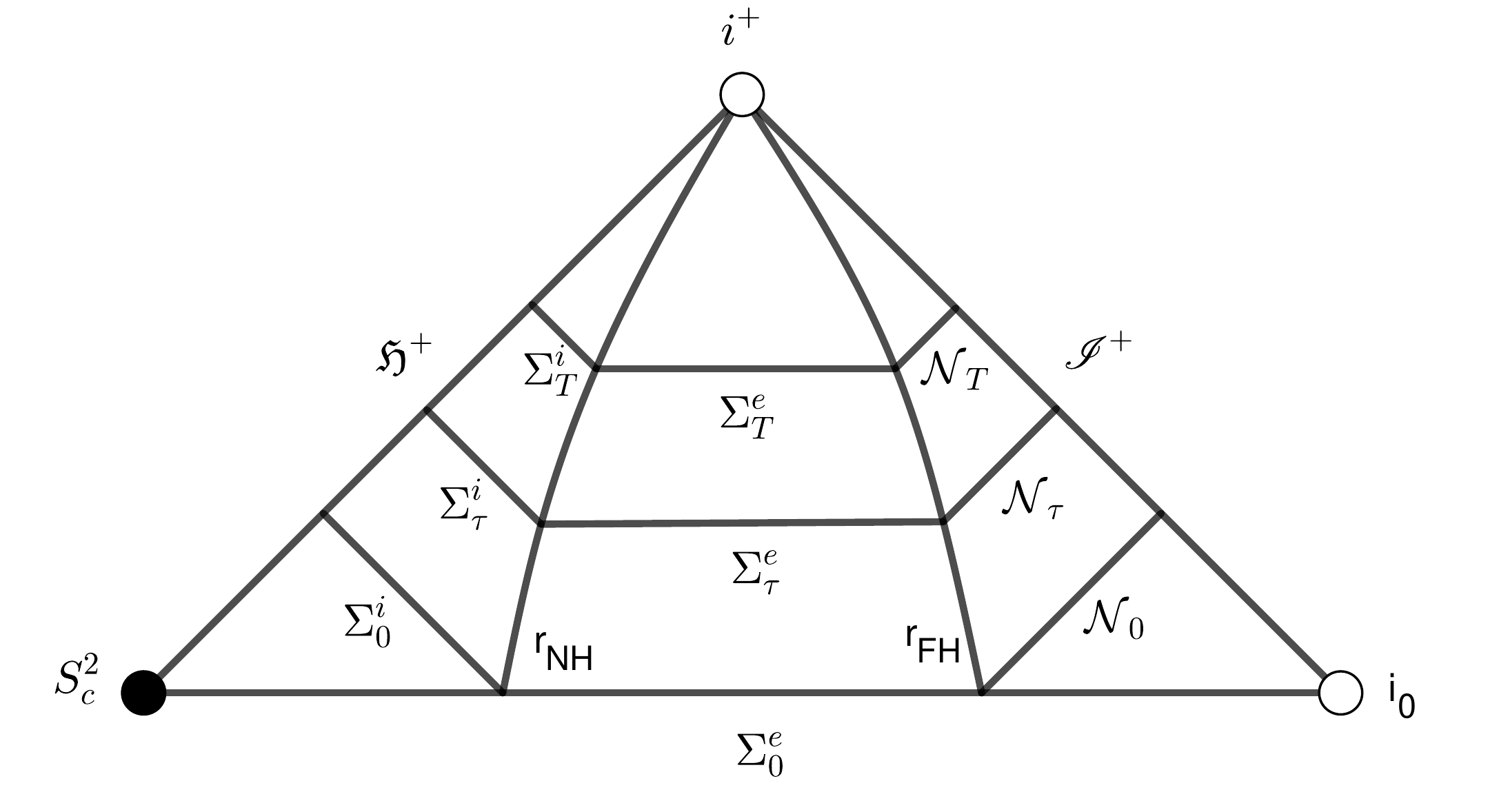}
\caption{Foliation $\left\{\mathcal{H}_\tau\right\}_{\tau\geq 0}$ of the Penrose's conformal compactification of Block I.}
\end{center}
\end{figure}
We known that $T=\partial_t$ is only time-like killing vector field on the exterior domain of the black hole $\left\{ r>2M \right\}$. To obtain the decay results for Equations \eqref{nonlinearequation} on the hypersurface $\mathcal{H}_\tau$, we need to choose and work on a global time-like vector field $K$ which is transverse to the horizon and equal to $\partial_t$ for $r$ large enough. The vector field $K$ is defined in $u,\,v,\,\omega$ coordinates by
\begin{align}\label{timelike}
K &= \partial_u + \partial_v + \frac{y_1(r)}{1-2M/r}\partial_u + y_2(r)\partial_v\nonumber\\
&= \left( 1+ \frac{y_1(r)}{F} \right)\partial_u + (1+y_2(r))\partial_v,
\end{align}
where $y_1, \, y_2>0$ are supported near the event horizon (where $r_+<r<r_{NH}$) and $y_1=1, \, y_2=0$ at the event horizon. 

\subsection{The equations}
The study of the equations governing the perturbations of the vacuum Schwarzschild
metric was initiated by Regge-Wheeler \cite{ReWhe} and then completed by Vishveshwara
\cite{Vis} and Zerilli \cite{Ze}. The perturbations of odd and even parity were treated
separately. The perturbations of odd parity are governed by the Regge-Wheeler
equation. Zerilli proved by decomposing into spherical harmonics (belonging to the different $l$ values), that the even parity perturbations are governed by the Zerilli equation. The relations between the Regger-Wheeler and Zerilli equations and the spin $\pm 2$ Teukolsky equations are detailed in \cite{ABlu2,Da2019,Da2020}.
 
In this paper, we will construct conformal scattering theories for the Regge-Wheeler and Zerilli equations and this is a first step to construct the conformal scattering theories of the linearized gravity fields.

The Regge-Wheeler and Zerilli equations have the following forms respectively
\begin{equation*}
\Box_g\psi - \frac{8M}{r^3}\psi =0
\end{equation*}
and
\begin{equation*}
\Box_g\psi - \frac{8M}{r^3} \frac{(2\bar{\lambda}+3)(2\bar{\lambda}r+3M)r}{4(\bar{\lambda}r+3M)^2}\psi =0,
\end{equation*}
where $2\bar{\lambda} = (l-1)(l+2)\geq 4$ (with $l\geq 2$) and the D'Alembertian for the Schwarzschild metric $g$ in variables $(t,\, r_*, \, \omega)$ is given by
$$\Box_g = \frac{1}{F}\left( \frac{\partial^2}{\partial t^2} - \frac{1}{r^2}\frac{\partial}{\partial r_*} r^2 \frac{\partial}{\partial r_*} \right) - \frac{1}{r^2}\Delta_{S^2}.$$
These equations are related by the Chandrasekhar transformation (see \cite{Cha1975,Cha}).

These two equations can be written together in a unique equation as
\begin{equation}\label{nonlinearequation}
\Box_g\psi + \frac{f(r,M)}{r^3}\psi = 0,
\end{equation}
where $f(r,M)$ is a negative, bounded lower and of order zero in $r$. In the Penrose's conformal compactification domain, the rescaled equations are
\begin{equation}\label{rescaledequation}
\Box_{\hat g}\hat{\psi} + (2M+f(r,M))R\hat{\psi} = 0,
\end{equation}
where $\hat{\psi} = \Omega^{-1}\psi=r\psi$ and the D'Alembertian for the rescaled metric $\hat{g}$ in variables $(t,\, r_*, \, \omega)$ is
$$\Box_{\hat g} = \frac{r^2}{F}\left( \frac{\partial^2}{\partial t^2} - \frac{\partial}{\partial r_*} \right) - \Delta_{S^2}.$$
We will construct a conformal scattering for the scalar fields which are the solutions of Equation \eqref{nonlinearequation} via the conformal approach to Equation \eqref{rescaledequation}.
We use the stress-energy tensor for the rescaled Equation \eqref{rescaledequation} as the one for the original linear wave equation
$$\hat{T}_{ab} = \hat{\nabla}_a\hat\psi \hat{\nabla}_b\hat\psi - \frac{1}{2}\left<\hat{\nabla}\hat\psi,\hat{\nabla}\hat\psi\right>_{\hat g}\hat{g}_{ab}.$$
The divergence of $\hat{T}$ is given by
$$\hat{\nabla}^a\hat{T}_{ab} = \left( \Box_{\hat g}\hat{\psi} \right)\hat{\nabla}_b\hat{\psi} = -(2M+f(r,M))R\hat{\psi} \hat{\nabla}_b(\hat\psi).$$
The energy current is obtained by contracting $\hat{T}$ with the time-like vector field $T=\partial_t$ i.e.,
$$\hat{J}_a = T^b\hat{T}_{ab}.$$
By taking $V= (M+f(r,M)/2)R \hat{\psi}^2\partial_t$, since $\hat{\Gamma}^{\bf a}_{\bf {a} 0}=0$ in the coordinate system $(t,\, r, \, \theta, \, \varphi)$, we have
$$\mathrm{div}V= \hat{\nabla}_aV^a = \frac{\partial}{\partial t}\left( (M+f(r,M)/2)R \hat{\psi}^2 \right) + \hat{\Gamma}^{\bf a}_{\bf {a} 0}V^0 = (2M+f(r,M))R\hat{\psi}\partial_t\hat\psi.$$
As a consequence we obtain the conservation law of $\hat{J}^a+V^a$ that is
\begin{equation}\label{conser}
\hat{\nabla}_a (\hat{J}^a + V^a) = 0, \hbox{ with } V= (M+f(r,M)/2)R \hat{\psi}^2\partial_t.
\end{equation}

\section{Energy identities}

\subsection{Energy identity up to $\mathcal{H}_T$}
We follow the convention used by by Penrose and Rindler \cite{PeRi84} about the Hodge dual of a 1-form $\alpha$ on a spacetime $({\cal M},g)$ (i.e. a $4-$dimensional Lorentzian manifold that is oriented and time-oriented)
\begin{equation*}
(*\alpha)_{abcd} = e_{abcd}{\alpha}^d,
\end{equation*}
where $e_{abcd}$ is the volume form on $({\cal M},g)$, denoted simply by $\mathrm{dVol}_g$. We shall use the following differential operator of the Hodge star
\begin{equation*}
\d *\alpha = -\frac{1}{4}(\nabla_a\alpha^a)\mathrm{{dVol}}_g.
\end{equation*}
If ${S}$ is the boundary of a bounded open set $\Omega$ and has outgoing
orientation, using Stokes theorem, we have
\begin{equation}\label{Stokesformula}
-4\int_{{S}}*\alpha = \int_{\Omega}(\nabla_a\alpha^a)\mathrm{{dVol}}_g.
\end{equation}
Therefore, for a solution $\hat\psi$ of the rescaled Equations \eqref{rescaledequation} with smooth and compactly supported initial data, we define the rescaled energy flux associated with $\partial_t$, across an oriented hypersurface $S$ by $\hat{\mathcal{E}}^{\partial_t}_S$:
$$\hat{\mathcal{E}}^{\partial_t}_S(\hat\psi) = -4\int_S *(\hat{J}_a + V_a)\d x^a = \int_S (\hat{J}_a+V_a)\hat{N}^a\hat{L}\hook \mathrm{{dVol}}_{\hat g},$$
where $\hat{L}$ is a transverse vector to $S$ and $\hat{N}$ is the normal vector field to $S$ such that
$\hat{L}^a\hat{N}_a=1$.
\begin{proposition}
Consider the smooth and compactly supported initial data on $\Sigma_0$, we can define the energy fluxes of the rescaled solution $\hat{\psi}$ across the null conformal boundary $\mathfrak{H}^+\cup \scri^+$ by $\hat{\mathcal{E}}^{\partial_t}_{\scri^+}(\hat\psi) + \hat{\mathcal{E}}^{\partial_t}_{\mathfrak{H}^+}(\hat\psi) := \lim_{T\rightarrow\infty}\hat{\mathcal{E}}^{\partial_t}_{\scri_T^+}(\hat\psi) + \hat{\mathcal{E}}^{\partial_t}_{\mathfrak{H}_T^+}(\hat\psi)$. Moreover, we have
$$\hat{\mathcal{E}}^{\partial_t}_{\scri^+}(\hat\psi) + \hat{\mathcal{E}}^{\partial_t}_{\mathfrak{H}^+}(\hat\psi) \leq \hat{\mathcal{E}}^{\partial_t}_{\Sigma_0}(\hat{\psi}),$$
where the equality holds if and only if $\lim_{T\rightarrow \infty}\hat{\mathcal{E}}^{\partial_t}_{\mathcal{H}_T}(\hat\psi) = 0$. 
\end{proposition} 
\begin{proof} 
Intergrating the conservation law \eqref{conser} and by using the Stokes's formula \eqref{Stokesformula}, we get an exact energy identity between the hypersurfaces $\Sigma_0=\left\{ t=0 \right\}, \, \mathfrak{H}_T^+, \, \mathcal{H}_T$ and $\scri_T^+$ (for $T>0$) as follows 
\begin{equation}\label{energyidentity}
\hat{\mathcal{E}}^{\partial_t}_{\Sigma_0}(\hat\psi)= \hat{\mathcal{E}}^{\partial_t}_{\scri_T^+}(\hat\psi) + \hat{\mathcal{E}}^{\partial_t}_{\mathfrak{H}_T^+}(\hat\psi) + \hat{\mathcal{E}}^{\partial_t}_{\mathcal{H}_T}(\hat\psi).
\end{equation}
On $\Sigma_0$, we take
$$\hat{L}=\frac{r^2}{F}\partial_t, \, \hat{N}=\partial_t.$$
On $\scri^+$, take $\hat{L}_{\scri^+}= -\partial_R$ in $u,R,\omega$ coordinates
$$\hat{L}_{\scri^+}= \frac{r^2F^{-1}}{2}l|_{\scri^+}.$$
On $\mathfrak{H}^+$, take $\hat{L}_{\mathfrak{H}^+} = \partial_R$ in $v,R,\omega$ coordinates
$$\hat{L}_{\mathfrak{H}^+} = \frac{r^2F^{-1}}{2}n|_{\mathfrak{H}^+}.$$
Hence we have $\hat{N}=\partial_t$ on both $\mathfrak{H}^+$ and $\scri^+$. This corresponds to $\hat{N}=\partial_v$ on $\mathfrak{H}^+$ and $\hat{N}=\partial_u$ on $\scri^+$. On the other hand, the vector field $V$ is parallel to $\partial_t$ which is null on $\mathfrak{H}$ and $\scri$, then $V_a\hat{N}^a=0$ in these two cases. The transversal and normal vectors of the hypersurface $\Sigma_T^i$ (resp. $\mathcal{N}_T$) can be choosen exactly as the ones of $\scri^+$ (resp. $\mathfrak{H}^+$). From this, it follows that the the energy identity given in \eqref{energyidentity} becomes
\begin{eqnarray*}
&&\int_{\mathfrak{H}_T^+}(\hat{J}_aT^a)\hat{L}_{\mathfrak{H}^+}\hook \mathrm{dVol}_{\hat g} + \int_{\scri_T^+}(\hat{J}_aT^a)\hat{L}_{\scri^+}\hook \mathrm{dVol}_{\hat g}\\
&&+ \int_{\Sigma_T^i}(\hat{J}_aT^a)\hat{L}_{\scri^+}\hook \mathrm{dVol}_{\hat g} + \int_{\mathcal{N}_T}(\hat{J}_aT^a)\hat{L}_{\mathfrak{H}^+}\hook \mathrm{dVol}_{\hat g} + \int_{\Sigma_T^e}((\hat{J}_a+V_a)T^a)r^2F^{-1}\partial_t \hook \mathrm{dVol}_{\hat g}\\
&&=\int_{\Sigma_0}((\hat{J}_a+V_a)T^a)r^2F^{-1}\partial_t \hook \mathrm{dVol}_{\hat g}\,.
\end{eqnarray*}
The energy fluxes through $\Sigma_0$, $\scri^+_T$ and $\mathfrak{H}^+_T$ can be computed in details. We have the following.
\begin{eqnarray*}
&&\hat{\mathcal{E}}^{\partial_t}_{\Sigma_0}(\hat\psi) = \frac{1}{2}\int_{\Sigma_0} \left( \hat{\psi_t}^2 + (\hat{\psi}_{r_*})^2 + R^2F|\nabla_{S^2}\hat{\psi}|^2 + (2M+f(r,M))FR^3\hat{\psi}^2 \right) \d r_*\d^2\omega,\\
&&\hat{\mathcal{E}}^{\partial_t}_{\scri^+_T} = \int_{\scri_T^+} (\partial_u(\hat{\psi}|_{\scri^+_T}))^2\d u\d^2\omega,\\
&&\hat{\mathcal{E}}^{\partial_t}_{\mathfrak{H}^+_T} = \int_{\mathfrak{H}^+} (\partial_v(\hat{\psi}|_{\mathfrak{H}^+_T}))^2\d v\d^2\omega.
\end{eqnarray*}
Note that in both cases of Regge-Wheeler and Zerilli, the energies through $\Sigma_0$ are positive due to the fact that we consider the parameter $l\geq 2$ which ensures that (for details see page 787 in \cite{ABlu}) 
$$\int_{S^2} \left( |\nabla_{S^2}\hat{\psi}|^2 + f(r,M)R \hat{\psi}^2 \right) \d^2 \omega \geq 0. $$
These calculations lead also to the energy flux through $\mathcal{H}_T= \Sigma_T^e \cup \Sigma_T^i \cup \mathcal{N}_T$ as 
\begin{eqnarray}\label{resenergy}
\mathcal{E}^{\partial_t}_{\mathcal{H}_T}(\hat\psi) &=& \mathcal{E}^{\partial_t}_{\Sigma_T^i}(\hat\psi) + \mathcal{E}^{\partial_t}_{\Sigma_T^e}(\hat\psi) + \mathcal{E}^{\partial_t}_{\mathcal{N}_T}(\hat\psi)\nonumber\\
&=& \int_{\Sigma_T^i} \left( (\partial_u(\hat{\psi}|_{\Sigma^i_T}))^2  +  R^2F |\nabla_{S^2}\hat\psi|_{\Sigma^i_T}|^2 + (2M+f(r,M))FR^3(\hat{\psi}|_{\Sigma^i_T})^2 \right)  \d u\d^2\omega \nonumber\\
&&+ \int_{\mathcal{N}_T} \left( (\partial_v(\hat{\psi}|_{\mathcal{N}_T}))^2  +  R^2F |\nabla_{S^2}\hat\psi|_{\mathcal{N}_T}|^2 + (2M+f(r,M))FR^3(\hat{\psi}|_{\mathcal{N}_T})^2 \right)  \d v\d^2\omega \nonumber\\
&&+ \frac{1}{2}\int_{\Sigma_T^e} \left( \hat{\psi_t}^2 + (\hat{\psi}_{r_*})^2 + R^2F|\nabla_{S^2}\hat{\psi}|^2 + (2M+f(r,M))FR^3\hat{\psi}^2 \right)\d r_*\d^2\omega.
\end{eqnarray}
We observe that the energy fluxes across $\scri_T^+$ and $\mathfrak{H}_T^+$ are non negative increasing functions of $T$ and their sum is bounded by $\hat{\mathcal{E}}_{\Sigma_0}(\hat\psi)$ by the energy indentity \eqref{energyidentity}. This can be deduced from the energy identity \eqref{energyidentity}and the positivity of $\hat{\mathcal{E}}^{\partial_t}_{\mathcal{H}_T}(\hat\psi)$. Therefore, the limit of $\hat{\mathcal{E}}^{\partial_t}_{\mathcal{H}_T}(\hat\psi)$ exists and the following sum is well defined
\begin{equation}\label{limitenergy}
\hat{\mathcal{E}}^{\partial_t}_{\scri^+}(\hat\psi) + \hat{\mathcal{E}}^{\partial_t}_{\mathfrak{H}^+}(\hat\psi) := \lim_{T\rightarrow\infty}\hat{\mathcal{E}}^{\partial_t}_{\scri_T^+}(\hat\psi) + \hat{\mathcal{E}}^{\partial_t}_{\mathfrak{H}_T^+}(\hat\psi) = \hat{\mathcal{E}}^{\partial_t}_{\Sigma_0}(\hat\psi) - \lim_{T\rightarrow\infty}\hat{\mathcal{E}}^{\partial_t}_{\mathcal{H}_T}(\hat\psi).
\end{equation}
The proposition now holds from the above identity.
\end{proof}

\subsection{Function space of initial data}
The relation between the energy fluxes $\hat{\mathcal E}^{\partial_t}_{\Sigma_T}(\hat\psi)$ and $\mathcal{E}^{\partial_t}_{\Sigma_T}(\psi)$ is shown in the following lemma.
\begin{lemma}\label{EnergyInvariant}
The energy fluxes $\hat{\mathcal E}^{\partial_t}_{\Sigma_T}(\hat\psi)$ and $\mathcal{E}^{\partial_t}_{\Sigma_T}(\psi)$ are equal.
\end{lemma}
\begin{proof}
We consider the stress-energy tensor for Equation \eqref{nonlinearequation} which is given by  
\begin{equation}\label{stressE}
T_{ab} = \nabla_a\psi \nabla_b\psi - \frac{1}{2}\left< \nabla \psi, \nabla\psi \right>_g g_{ab} + \frac{1}{2}\frac{f(r,M)}{r^3}\psi^2g_{ab},
\end{equation}
and which satisfies the conservation law $\nabla^aT_{ab} = 0$.
The energy flux associated $T=\partial_t$ through $\Sigma_0$ is
$$\mathcal{E}_{\Sigma_T}^{\partial_t}(\psi) = -4 \int_{\Sigma_0}*J_a \d x^a = \int_{\Sigma_T} J_aN^a L\hook \mathrm{dVol}_g,$$
where $J_a = T^bT_{ab}$ and $L=\dfrac{1}{F}\partial_t$ is transverse vector to $\Sigma_T$ and $N=\partial_t$ is a normal normal vector field to $\Sigma_T$ such that $L^aN_a=1$. By a straightforward calculation we get
\begin{eqnarray*}
\mathcal{E}^{\partial_t}_{\Sigma_T}(\psi) &=& \frac{1}{2} \int_{\Sigma_T} \left( (\hat{\psi}_t)^2 + (\hat{\psi}_{r_*})^2 + \frac{F}{r^2}|\nabla_{S^2}\hat\psi|^2  + \frac{FF'}{r}\hat{\psi}^2 + \frac{Ff(r,M)}{r^3}\hat{\psi}^2\right) \d r_*\d^2\omega \\
&=& \frac{1}{2} \int_{\Sigma_T} \left( (\hat{\psi}_t)^2 + (\hat{\psi}_{r_*})^2 + \frac{F}{r^2}|\nabla_{S^2}\hat\psi|^2  + (2M+f(r,M))FR^3\hat{\psi}^2 \right) \d r_*\d^2\omega\\
&=& \hat{\mathcal E}^{\partial_t}_{\Sigma_T}(\hat\psi).
\end{eqnarray*}
\end{proof}
\begin{remark}
\item[$\bullet$] The energy of the original fields across $\Sigma_T$ is defined by the stress-tensor $T_{ab}$ which is exactly the one defined in the work of Andersson et al. \cite{ABlu}. 
\item[$\bullet$] The equality $\hat{\mathcal E}^{\partial_t}_{\Sigma_T}(\hat{\psi}) = \mathcal{E}^{\partial_t}_{\Sigma_T}(\psi)$ does not mean that the energy is conformally invariant on the Cauchy hypersurface $\Sigma_T$. This is energy current $1$-form of the rescaled fields is obtained by contracting $\hat{J}^a+V^a$ with $\hat{T}_{ab}$ and the one of the original fields is obtained by contracting $T^a$ with $T_{ab}$.
\end{remark}
The finite energy space is defined as follows.
\begin{definition}
We denote by $\mathcal{H}$ the completion of $\mathcal{C}_0^\infty (\Sigma_0)\times \mathcal{C}_0^\infty (\Sigma_0)$ in the norm
$$\left\|(\hat{\psi}_0,\hat{\psi}_1) \right\|_{\mathcal H} = \frac{1}{\sqrt 2}\left(\int_{\Sigma_0} \left( (\hat{\psi}_t)^2 + (\hat{\psi}_{r_*})^2 + \frac{F}{r^2}|\nabla_{S^2}\hat\psi|^2  + (2M+f(r,M))FR^3\hat{\psi}^2 \right) \d r_*\d^2\omega \right)^{1/2}.$$
\end{definition}
Before stating the theorem of the well-posedness of Cauchy problem, we need the following definition of Sobolev spaces defined on open sets (see Definition 2 in \cite{Ni2016}).
\begin{definition}
Let $s\in [0,+\infty[$. A scalar function $u$ defined on $\mathcal{B}_I$ is said to belong to $H^s_{loc}(\bar{\mathcal{B}_I})$ if for any
local chart $(\Omega, \zeta)$, such that $\Omega \subset \mathcal{B}_I$ is an open set with smooth compact boundary in $\bar{\mathcal{B}_I}$ (note that this excludes neighbourhoods of either $i^\pm$ or $i^0$ but allows open sets whose boundary contains parts of the conformal boundary) and $\zeta$ is a smooth diffeomorphism from $\Omega$ onto a bounded open set $U \subset \mathbb{R}^4$ with smooth compact boundary, we have $u \circ \zeta^{-1} \in H^s(U)$.
\end{definition}
By Leray's theorem for the hyperbolic differential equations and the energy identity between $\Sigma_0$ and $\Sigma_t$, we have the following classic result.
\begin{proposition}
The Cauchy problem for Equation \eqref{nonlinearequation} on $\mathcal{B}_I$ (hence for the rescaled Equation\eqref{rescaledequation} on $\bar{\mathcal{B}_I}$) is well-posed in $\mathcal{H}$. This means that for any $(\hat{\psi}_0,\hat{\psi}_1) \in \mathcal{H}$, there exists a unique solution $\psi \in \mathcal{D}'(\mathcal{B}_I)$ of \eqref{nonlinearequation} such that
$$(r\psi,r\partial_t\psi) \in \mathcal{C}(\mathbb{R}_t;\mathcal{H}): \, r\psi|_{t=0}=\hat{\psi}_0; \, r\partial_t\psi|_{t=0} = \hat{\psi}_1.$$
Moreover, $\hat{\psi}=r\psi$ belongs to $H^1_{loc}(\bar{\mathcal B}_I)$.
\end{proposition}

\subsection{Energy identity up to $i^+$ and trace operator}
In this section, we will show that $\lim_{T\rightarrow \infty}\hat{\mathcal{E}}^{\partial_t}_{\mathcal{H}_T}(\hat\psi) = 0$ in \eqref{limitenergy} and obtain the relation between the energies  $\Sigma_0$, $\mathfrak{H}^+$ and $\scri^+$ by the following relation
$$\hat{\mathcal{E}}^{\partial_t}_{\Sigma_0}(\hat\psi) = \hat{\mathcal{E}}^{\partial_t}_{\mathfrak{H}^+}(\hat\psi) + \hat{\mathcal{E}}^{\partial_t}_{\scri^+}(\hat\psi).$$
We shall use the following results about the energy decay and pointwise decay of the solutions of the Regge-Wheeler and Zerilli Equations \eqref{nonlinearequation} which are obtained recently by Anderson, Blue and Wang (see Theorems 4.4 and 4.5 in \cite{ABlu})
\begin{theorem}i) (Energy Decay). Let $r_{FH}>3M$ be sufficiently large, and let $\psi$ be a solution to the Equations \eqref{nonlinearequation}, with the initial data imposed on $\mathcal{H}_{0}$ satisfying
$$\int_{\mathcal{N}_0}\sum_{k \leq 1}|T^k\partial_v(r\psi)|^2r^2\d v\d^2\omega + \int_{\mathcal{H}_0}\sum_{k \leq 2} J_{a}^K(T^k\psi)n^a_{\mathcal{H}_0}\d\mu_{\mathcal{H}_0} < \infty.$$
Then there exists a constant $I_1$ depending on the above initial data, such that
\begin{equation}\label{decay}
\mathcal{E}^K_{\mathcal{H}_T}(\psi)=\int_{\mathcal{H}_T} J_{a}^K(\psi)n^a_{\mathcal{H}_T}\d\mu_{\mathcal{H}_T} \leq \frac{I_1}{T^2}. 
\end{equation} 
ii) (Pointwise Decay). Denote $\Omega^l = \Omega_i^{i_1}\Omega_2^{i_2}\Omega_3^{i_3},\, i_1+i_2+i_3\leq l$, with $\Omega_i$ be three rotation Killing vector fields. Let $\psi$ be a solution to Equations \eqref{nonlinearequation}, with the initial data imposed on $\mathcal{H}_{0}$ satisfying
$$\sum_{l\leq 2}\int_{\mathcal{N}_0}\sum_{k \leq 1}|T^k\Omega^l\partial_v(r\psi)|^2r^2 \d v\d^2\omega + \int_{\mathcal{H}_0}\sum_{k \leq 2} J_{a}^K (T^k\Omega^l\psi)n^a_{\mathcal{H}_0}\d\mu_{\mathcal{H}_0} < \infty.$$
Then there exists a constant $I_2$ depending on the above data, such that in the future development of the initial hypersurface $\mathcal{J}^+(\mathcal{H}_0)$ we have
\begin{equation}\label{pointwisedecay}
r^{1/2}|\psi(T,r)| \leq \frac{I_2}{T}.
\end{equation}
Here $J_{a}^K(\psi)= K^bT_{ab}(\psi)$ is the contraction of the stress-energy tensor $T_{ab}$ with the global time-like vector field $K$ (given by \eqref{timelike}), $n^a_S$ is the future normal vector to the hypersurface $S$ and $\d\mu_{S}$ is the volume form of $S$. 
\end{theorem}
Now we state and prove the main theorem of this section.
\begin{theorem}\label{EQ}
let $\psi$ be a solution to Equations \eqref{nonlinearequation}, with the initial data imposed on $\mathcal{H}_{0}$ satisfying
$$\sum_{l\leq 2}\int_{\mathcal{N}_0}\sum_{k \leq 1}|T^k\Omega^l\partial_v(r\psi)|^2r^2 \d v\d^2\omega + \int_{\mathcal{H}_0}\sum_{k \leq 2} J_{a}^K (T^k\Omega^l\psi)n^a_{\mathcal{H}_0}\d\mu_{\mathcal{H}_0} < \infty.$$
The energy of the rescaled fields $\hat{\psi}$ across the hypersurface $\mathcal{H}_T$ tends to zero as $T$ tends to infinity
$$\lim_{T\rightarrow \infty}\hat{\mathcal{E}}^{\partial_t}_{\mathcal{H}_T}(\hat\psi) = 0.$$
As a consequence, the equality of the energies holds true
\begin{equation}\label{energyequality}
\hat{\mathcal{E}}^{\partial_t}_{\Sigma_0}(\hat\psi) = \hat{\mathcal{E}}^{\partial_t}_{\mathfrak{H}^+}(\hat\psi) + \hat{\mathcal{E}}^{\partial_t}_{\scri^+}(\hat\psi).
\end{equation}
\end{theorem}
\begin{proof}
Since $K=\partial_t$ on the domain $r\geq r_{NH}$ and from the energy property in Lemma \ref{EnergyInvariant}, we obtain that
$$\hat{\mathcal{E}}^K_{\Sigma_T^e}(\hat\psi) = \hat{\mathcal{E}}^{\partial_t}_{\Sigma_T^e}(\hat\psi)=\mathcal{E}^{\partial_t}_{\Sigma_T^e}(\psi) = \mathcal{E}^K_{\Sigma_T^e}(\psi).$$
We recall that $J_a = (\partial_t)^bT_{ab} = \frac{1}{2}(\partial_u + \partial_v)^b T_{ab}$ with $T_{ab}$ is given by \eqref{stressE} i.e, 
$$T_{ab} = \nabla_a\psi \nabla_b\psi - \frac{1}{2}\left< \nabla \psi, \nabla\psi \right>_g g_{ab} + \frac{1}{2}\frac{f(r,M)}{r^3}\psi^2g_{ab},$$ 
$g_{ab} = F (\d u^2 - \d v^2) - r^2\d \omega^2$ and $\mathrm{dVol}_g = \frac{r^2F}{2}\d u \wedge \d v \wedge \d^2\omega$.
The energy flux across $\mathcal{N}_T$ can be calculated by taking 
$$L|_{\mathcal{N}_T} = \frac{F^{-1}}{2}n|_{\mathcal{N}_T} \hbox{  and  } N|_{\mathcal{N}_T} = \partial_v.$$
We get
\begin{align*}
\mathcal{E}^{K}_{\mathcal{N}_T}(\psi)&= \mathcal{E}^{\partial_t}_{\mathcal{N}_T}(\psi)= \int_{\mathcal{N}_T}(J_aN^a)L|_{\mathcal{N}_T}\hook \mathrm{dVol}_g \\
&=\int_{\mathcal{N}_T} \left( |\partial_v (\psi|_{\mathcal{N}_T})|^2 + F\left(\frac{1}{r^2} |\nabla_{S^2}(\psi|_{\mathcal{N}_T})|^2 + \frac{f(r,M)}{r^3}|\psi|_{\mathcal{N}_T}|^2 \right) \right) r^2\d v \d^2\omega.
\end{align*}
Now we have
$$J^K_a = K^bT_{ab} = \left[ (1+y_1(r)F^{-1})\partial_u + (1+y_2(r))\partial_v \right]^b T_{ab},$$
where $y_1, \, y_2>0$ are supported near the event horizon (where $r_+<r<r_{NH}$) and $y_1=1, \, y_2=0$ at the event horizon.
The energy flux across $\Sigma_T^i$ can be calculated in variables $u,\, r_*,\, \omega$ by taking
$$L|_{\Sigma_T^i} = \frac{F^{-1}}{2}l|_{\Sigma_T^i} \hbox{  and  } N|_{\Sigma_T^i} = \partial_u.$$
More precisely, we have
\begin{align*}
\mathcal{E}^K_{\Sigma_T^i}(\psi) &= \int_{\Sigma_T^i}(J^K_aN^a)L|_{\Sigma_T^i}\hook \mathrm{dVol}_g \\
&= \int_{\Sigma_T^i} \left( (1+y_1(r)F^{-1})|\partial_u \psi|_{\Sigma_T^i}|^2 + F \left( \frac{1}{r^2}|\nabla_{S^2}\psi|_{\Sigma_T^i}|^2 + \frac{f(r,M)}{r^3}|\psi|_{\Sigma_T^i}|^2 \right) \right)r^2 \d u \d^2\omega\\
&\geq \int_{\Sigma_T^i} \left( |\partial_u \psi|_{\Sigma_T^i}|^2 + F \left( \frac{1}{r^2}|\nabla_{S^2}\psi|_{\Sigma_T^i}|^2 + \frac{f(r,M)}{r^3}|\psi|_{\Sigma_T^i}|^2 \right) \right)r^2 \d u \d^2\omega.
\end{align*}
Therefore we obtain that
\begin{eqnarray*}
{\mathcal{E}}^K_{\mathcal{H}_T}(\psi) &=& {\mathcal{E}}^K_{\Sigma_T^i}(\psi) + {\mathcal{E}}^K_{\Sigma_T^e}(\psi)+ {\mathcal{E}}^K_{\mathcal{N}_T}(\psi)\\
&\geq& \int_{\Sigma_T^i} \left( |\partial_u \psi|_{\Sigma_T^i}|^2 + F \left( \frac{1}{r^2}|\nabla_{S^2}\psi|_{\Sigma_T^i}|^2 + \frac{f(r,M)}{r^3}|\psi|_{\Sigma_T^i}|^2 \right) \right)r^2 \d u \d^2\omega + {\mathcal E}^{K}_{\Sigma_T^e}(\psi)  \nonumber\\
&&+ \int_{\mathcal{N}_T} \left( |\partial_v (\psi|_{\mathcal{N}_T})|^2 + F\left(\frac{1}{r^2} |\nabla_{S^2}(\psi|_{\mathcal{N}_T})|^2 + \frac{f(r,M)}{r^3}|\psi|_{\mathcal{N}_T}|^2 \right) \right) r^2\d v \d^2\omega \nonumber\\
&\geq&\int_{\Sigma_T^i} \left( |\partial_u (\hat{\psi}|_{\Sigma_T^i}) + F\psi|_{\Sigma_T^i}|^2 + F\left(R^2 |\nabla_{S^2}(\hat{\psi}|_{\Sigma_T^i})|^2 + f(r,M)R^3|\hat\psi|_{\Sigma_T^i}|^2 \right) \right)\d u \d^2\omega \\
&&+ \hat{\mathcal E}^{\partial_t}_{\Sigma_T^e}(\psi) \nonumber\\
&&+ \int_{\mathcal{N}_T} \left( |\partial_v (\hat\psi|_{\mathcal{N}_T}) - F\psi|_{\mathcal{N}_T}|^2 +  F\left(R^2 |\nabla_{S^2}(\hat\psi|_{\mathcal{N}_T})|^2 + f(r,M)R^3|\hat\psi|_{\mathcal{N}_T}|^2 \right)  \right)\d v \d^2\omega\\
&\geq&\int_{\Sigma_T^i} \left( |\partial_u (\hat{\psi}|_{\Sigma_T^i})|^2 - F^2|\psi|_{\Sigma_T^i}|^2 \right.\\
&&\hspace{3cm} \left. + F\left(R^2 |\nabla_{S^2}(\hat{\psi}|_{\Sigma_T^i})|^2 + (2M+f(r,M))R^3|\hat\psi|_{\Sigma_T^i}|^2 \right) \right)\d u \d^2\omega \\
&&+ \hat{\mathcal E}^{\partial_t}_{\Sigma_T^e}(\psi) - \int_{\Sigma_T^i} 2MR^3 |\hat{\psi}|_{\Sigma_T^i}|^2\d u \d^2\omega - \int_{\mathcal{N}_T} 2MR^3 |\hat{\psi}|_{\mathcal{N}_T}|^2\d v \d^2\omega \nonumber\\
&&+ \int_{\mathcal{N}_T} \left( |\partial_v (\hat\psi|_{\mathcal{N}_T})|^2 - F^2|\psi|_{\mathcal{N}_T}|^2 \right.\\
&&\hspace{3cm}\left. +  F\left(R^2 |\nabla_{S^2}(\hat\psi|_{\mathcal{N}_T})|^2 + (2M+f(r,M))R^3|\hat\psi|_{\mathcal{N}_T}|^2 \right)  \right)\d v \d^2\omega.
\end{eqnarray*}
Together with\eqref{resenergy}, \eqref{decay} and \eqref{pointwisedecay}, the above equation yields
\begin{eqnarray*}
\hat{\mathcal{E}}^{\partial_t}_{\mathcal{H}_T}(\hat\psi) &\leq& {\mathcal{E}}^K_{\mathcal{H}_T}(\psi)+ \int_{\Sigma_T^i}(F^2+2MR) |\psi|^2 \d u \d^2\omega +  \int_{\mathcal{N}_T} (F^2+2MR)|\psi|^2 \d v \d^2\omega\\
&\leq& \frac{I_1}{T^2} +  \left(\int_{\Sigma_T^i}\frac{I^2_2}{rt^2} \d u \d^2\omega +  \int_{\mathcal{N}_T}\frac{I^2_2}{rt^2} \d v \d^2\omega \right)\\
&\leq& \frac{I_1}{T^2} +  (F^2+2)\left( \int_{T-(r_{NH})_*}^\infty\int_{S^2}\frac{I_2^2}{2Mt^2} \d u \d^2\omega + \int_{T+(r_{FH})_*}^\infty\int_{S^2}\frac{I_2^2}{r_{FH}t^2}\d v \d^2\omega \right)\\
&\leq& \frac{I_1}{T^2} +  (F^2+2) \left( \int_{T-(r_{NH})_*}^\infty\int_{S^2}\frac{4I_2^2}{2M(u+v)^2} \d u \d^2\omega + \int_{T+(r_{FH})_*}^\infty\int_{S^2}\frac{4I_2^2}{r_{FH}(u+v)^2} \d v \d^2\omega \right)\\
&\leq& \frac{I_1}{T^2} + (F^2+2)\left( \frac{2\pi I_2^2}{M}\frac{1}{T-(r_{NH})_* + v|_{\Sigma^i_T}} + \frac{4\pi I^2_2}{r_{FH}}\frac{1}{T+(r_{FH})_* + u|_{\mathcal{N}_T}} \right).
\end{eqnarray*}
Therefore, 
$$\lim_{T\rightarrow \infty}\hat{\mathcal{E}}^{\partial_t}_{\mathcal{H}_T}(\hat\psi) = 0.$$
Plugging into \eqref{limitenergy} we obtain the energy identity up to $i^+$, i.e, the relation between the energy across $\Sigma_0$ and the ones across $\mathfrak{H}^+\cup \scri^+$ as follows 
\begin{equation*}
\hat{\mathcal{E}}^{\partial_t}_{\Sigma_0}(\hat\psi) = \hat{\mathcal{E}}^{\partial_t}_{\mathfrak{H}^+}(\hat\psi) + \hat{\mathcal{E}}^{\partial_t}_{\scri^+}(\hat\psi).
\end{equation*}
\end{proof}
\begin{definition}(Trace operator). Let $(\hat{\psi}_0, \hat{\psi}_1)\in \mathcal{C}_0^\infty(\Sigma_0)\times \mathcal{C}_0^\infty(\Sigma_0)$. Consider the solution of Equations \eqref{rescaledequation}, $\hat{\psi} \in \mathcal{C}^\infty(\bar{\mathcal{B}_I})$ such that
$$\hat{\psi}|_{\Sigma_0} = \hat{\psi}_0, \, \partial_t \hat{\psi}_1|_{\Sigma_0} = \hat{\psi}_1.$$
We define the trace operator $\mathcal{T}^+$ from $\mathcal{C}_0^\infty(\Sigma_0)\times \mathcal{C}_0^\infty(\Sigma_0)$ to $\mathcal{C}^\infty(\mathfrak{H}^+)\times \mathcal{C}_0^\infty(\scri^+)$ as follows
$$\mathcal{T}^+(\hat{\psi}_0, \hat{\psi}_1) = (\hat{\psi}|_{\mathfrak{H}^+}, \hat{\psi}|_{\scri^+}).$$
\end{definition}
\begin{remark}
The energy identity up to $i^+$ \eqref{energyequality} shows that the trace operator $\mathcal{T}^+$ is injective.
\end{remark}
\begin{definition}
The function space for scattering data $\mathcal{H}^+$ is the completion of $\mathcal{C}_0^\infty(\mathfrak{H}^+) \times \mathcal{C}_0^\infty(\scri^+)$ in the norm
$$\left\| (\xi,\zeta) \right\|_{\mathcal{H}^+} = \left(\int_{\mathfrak{H}^+}(\partial_v\xi)^2 \d v\d^2\omega + \int_{\scri^+}(\partial_u\zeta)^2 \d u \d^2\omega \right)^{1/2}.$$
This means that
$$\mathcal{H}^+ \simeq \dot{H}^1(\mathbb{R}_v; L^2(S^2_\omega)) \times \dot{H}^1(\mathbb{R}_u; L^2(S^2_\omega)).$$
\end{definition}
\begin{theorem}\label{Trace}
The trace operator extends uniquely as a bounded linear map from $\mathcal{H}$ to $\mathcal{H}^+$. The extended operator is a partial isometry i.e. for any $(\hat{\psi}_0,\hat{\psi}_1) \in \mathcal{H}$,
$$\left\| \mathcal{T}^+(\hat{\psi}_0,\hat{\psi}_1) \right\|_{\mathcal{H}^+} = \left\| (\hat{\psi}_0,\hat{\psi}_1) \right\|_{\mathcal{H}}.$$
\end{theorem}
\begin{proof}
The proof of this theorem is a direct consequence of the equality energy \eqref{energyequality} obtained in Theorem \ref{EQ}.
\end{proof}

\section{Conformal scattering theory}
\subsection{Applying the result of H\"ormander in the Schwarzschild background}
H\"ormander \cite{Ho1990} solved the Goursat problem for the second order scalar wave equation with regular first order potentials in the spatially compact spacetime. 
Nicolas \cite{Ni2006} extended the results of H\"ormander with very slightly more regular metric and potential, precisely a $\mathcal{C}^1$ metric and a potential with continuous coefficients of the first order terms and locally $L^\infty$ coefficients for the terms of order $0$.
Here we will apply the results in \cite{Ho1990,Ni2006} to solve the Goursat problem for Equation \eqref{rescaledequation}: 
$$\Box_{\hat{g}}\hat{\psi} + ( 2MR+ f(r,M) )\hat{\psi} = 0,$$
with the smoothtly supported compact initial data on the conformal boundary i.e, $(\xi, \zeta) \in \mathcal{C}_0^\infty(\mathfrak{H}^+) \times \mathcal{C}_0^\infty(\scri^+)$ in Schwarzschild background.

By the same manner in Appendix B \cite{Ni2016}, the conformal compactification domain $\bar{\mathcal{B}}_I$ is embedded partly in a cylindrical globally hyperbolic spacetime by the following technique: let $\mathcal{S}$ be a spacelike hypersurface on $\bar{\mathcal{B}}_I$ such that its intersection with the horizon is the crossing sphere and which crosses $\scri^+$ strictly in the past of the support of the data. We cut off the future $\mathcal{V}$ of a point in $\bar{\mathcal{B}}_I$ lying in the future of the past of the support of the Goursat data and get a spacetime denoted by $\mathfrak{M}$. Then we extend $\mathfrak{M}$ as a cylindrical globally hyperbolic spacetime $(\mathbb{R}_t\times S^3, \mathfrak{g})$. The conformal boundaries $\mathfrak{H}^+\cup \scri^+$ is extended inside $\mathcal{I}^+(\mathcal{S})/\mathcal{V}$ as a null hypersurface $\mathcal{C}$, that is the graph of a Lipschitz function over $S^3$ and the data by zero on the rest of the extended hypersurface.

\begin{figure}[htb]
\begin{center}
\includegraphics[scale=0.5]{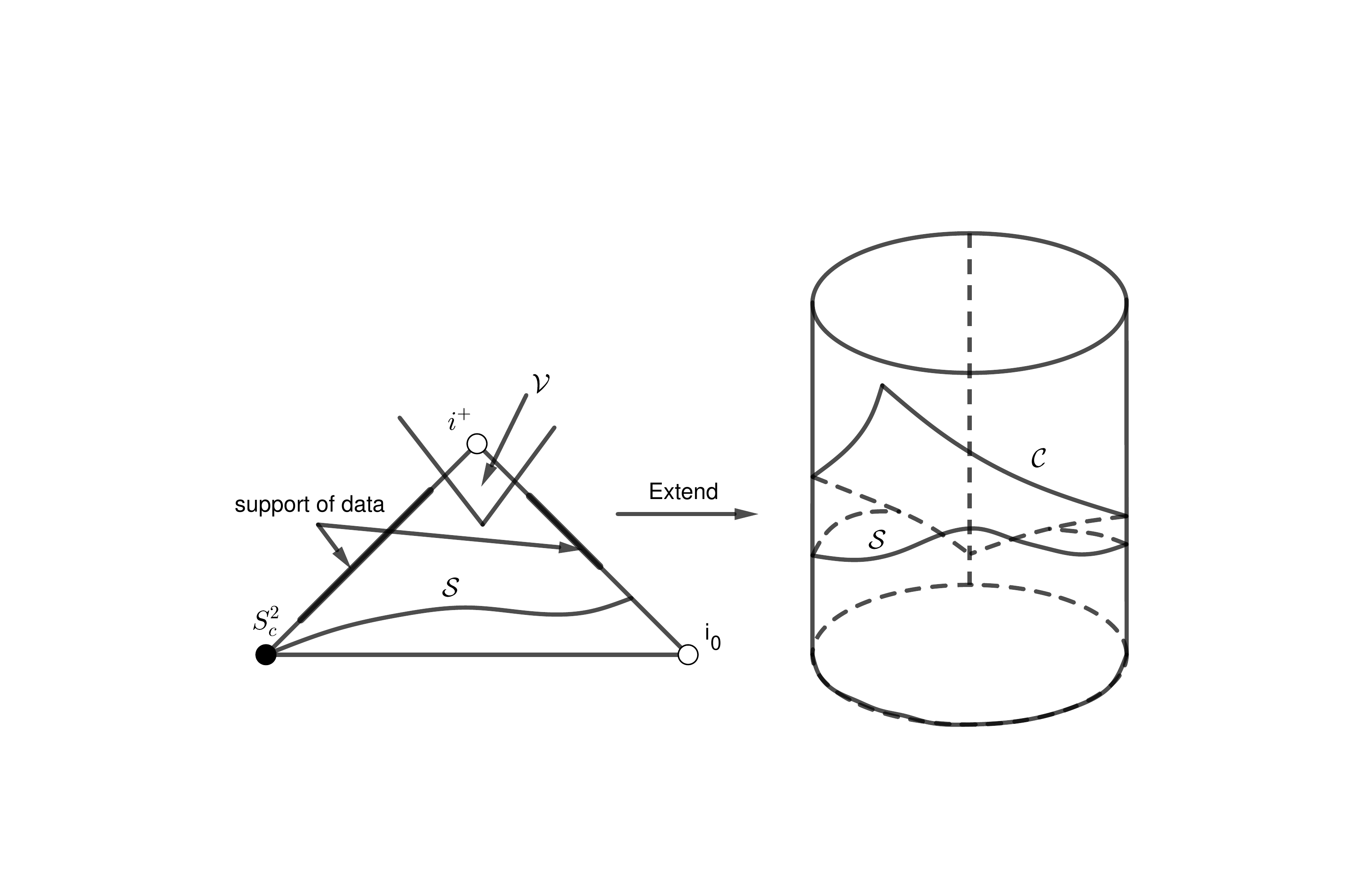}
\caption{The cutting off $\mathcal{V}$ and extension of the future $\mathcal{I}^+(\mathcal{S})$.}
\end{center}
\end{figure}

In the spacetime $(\mathbb{R}_t\times S^3, \mathfrak{g})$, the Equation \eqref{rescaledequation} becomes
\begin{equation}\label{extendeq}
\Box_{\mathfrak{g}}\hat{\psi} + ( 2MR+ f(r,M) )\hat{\psi} = 0.
\end{equation}
Since the function $f(r,M)$ is of order zero in $r$, it is bounded. Therefore, we can apply the results of H\"ormander \cite{Ho1990} to obtain the well-posedness of the Goursat problem of Equation \eqref{extendeq}. 

By local uniqueness and causality, using in particular the fact that as a consequence of the
finite propagation speed, the solution $\psi$ vanishes in $\mathcal{I}^+(\mathcal{S})/\bar{\mathcal{B}}_I$, the Goursat problem of Equation \eqref{rescaledequation} has a unique smooth solution in the future of $\mathcal{S}$, that is the restriction of $\psi$ to $\mathfrak{M}$. Therefore, we have established the following theorem.
\begin{theorem}\label{partlyGoursat}
With the initial data $(\xi, \, \zeta) \in {\mathcal C}_0^{\infty}(\mathfrak{H}^+)\times {\mathcal{C}}_0^\infty(\scri^+)$, the Equation \eqref{rescaledequation} has a unique solution $\hat{\psi}$ satisfying
$$\hat{\psi} \in {\mathcal {C}}(\mathbb{R};\, H^1(\mathcal{S}))\cap \mathcal{C}^1(\mathbb{R}; \, L^2(\mathcal{S}))$$
and
$$\hat{\psi}|_{\scri^+} = \zeta, \, \hat{\psi}|_{\mathfrak{H}^+}= \xi\,.$$
\end{theorem}
\begin{remark}
We have solved the partly Goursat problem for Equation \eqref{rescaledequation} in the future $\mathcal{I}^+(\mathcal{S})$ of the spacelike hypersurface $\mathcal{S}$. The Goursat problem will be solved completely by extending the solution down to $\Sigma_0$, this will be done in the next section.
\end{remark}

\subsection{Goursat problem and conformal scattering operator}
In order to obain the conformal scattering operator, we need to show that the trace operator is subjective. This corresponds to solve the Goursat problem for Equation \eqref{rescaledequation} with the initial data on the conformal boundary $\mathfrak{H}^+\cup \scri^+$ (resp. $\mathfrak{H}^-\cup \scri^-$) in the conformal compactification domain $\bar{\mathcal{B}}_I$.
\begin{theorem}
The Goursat problem of Equation \eqref{rescaledequation} is well-posed i.e. for the initial data $(\xi,\zeta)\in C_0^\infty(\mathfrak{H}^+)\times C_0^\infty(\scri^+)$, there exists a unique solution of \eqref{rescaledequation} such that
$$(\hat\psi,\partial_t\hat\psi)\in \mathcal{C}(\mathbb{R}_t;\mathcal{H}) \hbox{  and  } \mathcal{T}^+(\hat\psi|_{\Sigma_0}, \partial_t \hat\psi|_{\Sigma_0}) = (\xi,\zeta).$$
This means that the trace operator $\mathcal{T}^+: \mathcal{H} \to \mathcal{H}^+$ is subjective and and together with Theorem \ref{Trace}, this yields that the trace operator $\mathcal{T}^+$ is an isometry.
\end{theorem}
\begin{proof}
Following Theorem \ref{partlyGoursat} there exists a unique solution $\hat{\psi}$ of \eqref{rescaledequation} which satisfies that
\begin{itemize}
\item[$\bullet$] $\hat{\psi} \in H^1(\mathcal{I}^+(\mathcal{S}))$, where $\mathcal{I}^+(\mathcal{S})$ is the causal future of $\mathcal{S}$ in $\bar{\mathcal{B}_I}$. Since the support of the initial data is compact, the solution vanishes in the neigbourhood of $i^+$. Then we do not need to distinguish between $H^1(\mathcal{I}^+(\mathcal{S}))$ and $H^1_{loc}(\mathcal{I}^+(\mathcal{S}))$. 

\item[$\bullet$] $\hat\psi$ is continuous in $\tau$ with values in $H^1$ of the slices and $C^1$ in $\tau$ with values in $L^2$ of the slices.

\item[$\bullet$] $\hat{\psi}|_{\scri^+} = \zeta, \, \hat{\psi}|_{\mathfrak{H}^+}=\xi$.
\end{itemize}

\begin{figure}[htb]
\begin{center}
\includegraphics[scale=0.5]{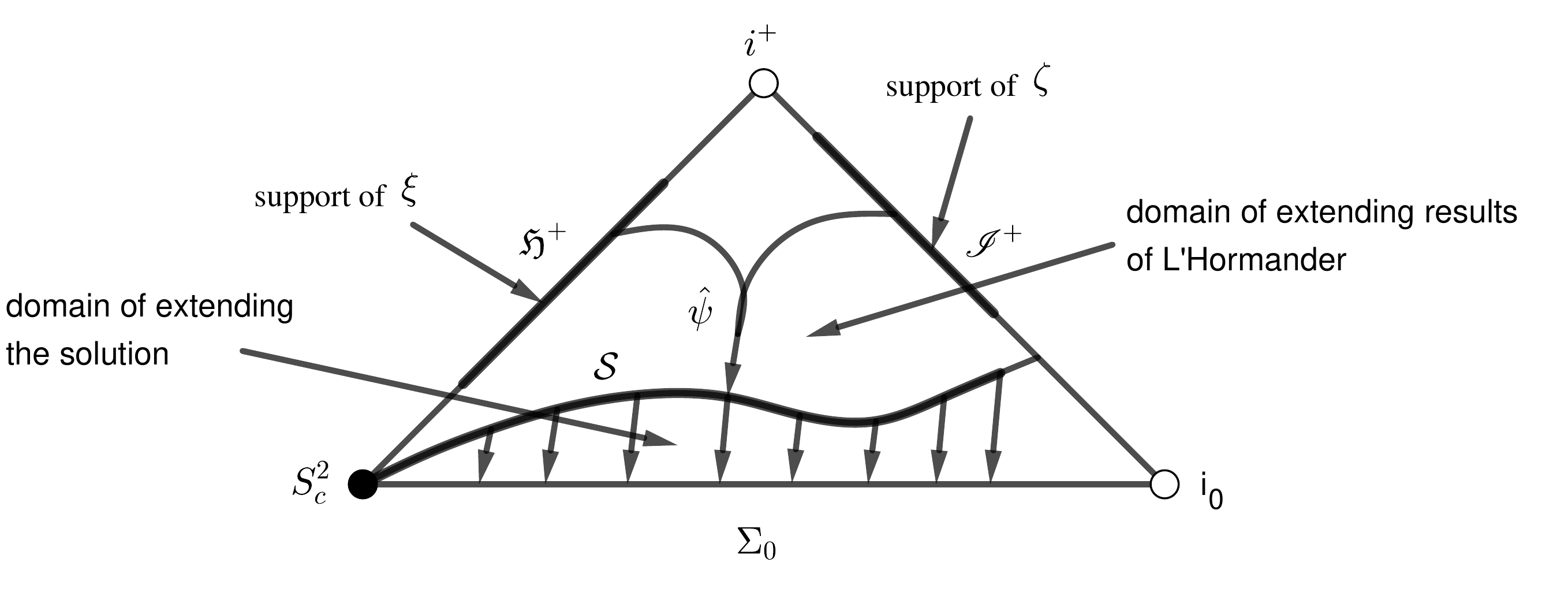}
\caption{Extending the solution $\hat{\psi}$ down to $\Sigma_0$.}
\end{center}
\end{figure}

We need to extend the solution down to $\Sigma_0$ in a manner that avoids the singularity at $i_0$.
Since the hypersurface $\mathcal{S}$ intersects the horizon at the crossing sphere and intersects $\scri^+$ strictly in the past of the support of the data, we have that the restriction of $\hat{\psi}$ to $\mathcal{S}$ is in $H^1(\mathcal{S})$ and its trace on $\mathcal{S}\cap \scri^+$ is also the trace of $\zeta$ on $\mathcal{S}\cap \scri^+$. Therefore, we have that $\hat{\psi}|_{\mathcal{S}}$ can be approached by a sequence $\left\{\hat{\psi}^n_{0,\mathcal{S}} \right\}_{n\in \mathbb{N}}$ of the smooth functions on $\mathcal{S}$ supported away from $\scri^+$ that converge towards $\hat{\psi}|_{\mathcal{S}}$ in $H^1(\mathcal{S})$. Moreover, $\partial_t\hat{\psi}|_{\mathcal{S}}$ can be approached by a sequence $\left\{\hat{\psi}^n_{1,\mathcal{S}} \right\}_{n\in \mathbb{N}}$ of the smooth functions on $\mathcal{S}$ supported away from $\scri^+$ that converge towards $\partial_t\hat{\psi}|_{\mathcal{S}}$ in $L^2(\mathcal{S})$.
Consider $\hat{\psi}^n\,,$ the smooth solution of \eqref{rescaledequation} on $\bar{\mathcal{B}_I}$ with data $(\hat{\psi}^n_{0,\mathcal{S}},\hat{\psi}^n_{1,\mathcal{S}})$ on $\mathcal{S}$. This solution vanishes in the neighbourhood of $i_0$ and we can therefore perform energy estimates for $\hat{\psi}^n$ between $\mathcal{S}$ and $\Sigma_0$, then we have the energy identity
\begin{equation}
\mathcal{E}^{\partial_t}(\mathcal{S},\hat{\psi}^n) = \mathcal{E}^{\partial_t}(\Sigma_0,\hat{\psi}^n).
\end{equation}
Similar energy identities between $\mathcal{S}$ and the hypersurfaces $\Sigma_t$ entail that $(\hat{\psi}^n,\partial_t\hat{\psi}^n)$ converges in $\mathcal{C}(\mathbb{R}_t,\mathcal{H})$ towards $(\hat{\phi}, \partial_t\hat{\phi})$, where $\hat{\phi}$ is a solution of \eqref{rescaledequation}. By local uniqueness $\hat{\phi}$ coincides with $\hat{\psi}$ in the future of $\mathcal{S}$. Therefore, if we denote
$$\hat{\psi}_0 = \hat{\psi}|_{\Sigma_0},\,\, \hat{\psi}_1 = \partial_t\hat{\psi}|_{\Sigma_0},$$
we have
$$(\hat{\psi}_0,\hat{\psi}_1)\in \mathcal{H} \hbox{  and  } (\xi,\zeta) = \mathcal{T}^+(\hat{\psi}_0,\hat{\psi}_1).$$
Therefore, the range of $\mathcal{T}^+$ contains $\mathcal{C}_0^\infty(\mathfrak{H}^+)\times \mathcal{C}_0^\infty(\scri^+)$ and $\mathcal{T}^+$ is subjectie.
\end{proof}
We now define the conformal scattering operator for the Regge-Wheeler and Zerilli equations as follows.
\begin{definition}
Similarly, we introduce the past trace operator $\mathcal{T}^-$ and the space $\mathcal{H}^-$ of past scattering data on the past horizon and the past null infinity. We define the scattering operator $S$ as the operator that, to the past scattering data associates the future scattering data, i.e.
$$S:= \mathcal{T}^+\circ (\mathcal{T}^-)^{-1}.$$
\end{definition}

\begin{remark}

\item[$\bullet$] Before the first version of this work on ARXIV, there was not an analytic scattering theory for the linearized gravity fields on Schwarzschild spacetime. 
Conformal and analytic scattering are two different approaches to scattering theory. The first approach is based on spectral techniques and the last one uses the Penrose conformal compactification and energy method. The main difference between two approaches is that conformal scattering gives the scattering construction by solving the Goursat problem on the conformal boundary while analytic analytic scattering considers the scattering as asymptotic regions. One can understand the relation between the two approaches by re-interpreting the conformal scattering by analytic scattering (this work is done for linear scalar fields in Section 4.2 in \cite{Ni2016} ). Since that, our construction on the conformal scattering theory for the linearized gravity fields can also lead to understand the analytic aspect of the scattering theory that is defined in terms of wave operator.

\item[$\bullet$] The extensions of this work on the other symmetric spherical spacetimes such as Reissner-Nordstr\"om-de Sitter black hole can be done by the same way: first, extend the work of Anderson et al. \cite{ABlu} to obtain the decay results (where the recent works \cite{El2020,El2020'} of Giorgi can be useful) and then using these results to the construction, where it remains useful to use the time-like Killing vector field $\partial_t$ to obtain the energy of the fields. However, the extension on the stationary and non symmetric spherical spacetimes such as Kerr spacetimes is more complicated. In Kerr spacetimes, the existence of the orbiting null geodesics and the fact that the vector field $\partial_t$ is no longer global time-like in the exterior of the black hole lead an issue that the conserved energies on the Cauchy hypersurfaces of the spacetime is not defined as usual. Some recent results about the uniform bounded energy and the pointwise decay of the fields satisfying the Teukolsky equations of Andersson et al. \cite{ABlu2} and also Dafermos et al. \cite{Da2020} can be useful. 

\item[$\bullet$] Equations \eqref{rescaledequation} are very similar to the rescaled wave equations, hence they also satisfy the peeling property (which is another aspect of the conformal asymptotic analysis). The method is done by the same way for the linear scalar fields on Schwarzschild spacetimes in \cite{MaNi2009} and extended to Kerr spacetimes in \cite{NiXu2019,Pha2019}.

\end{remark}


\begin{thebibliography}{100}


\bibitem{ABlu} L. Andersson, P. Blue and J. Wang, {\em Morawertz estimate for linearized gravity in Schwarzschild}, Ann. Henri Poincar\'e 21 (2020), 761-813, arXiv:1708.06943.

\bibitem{ABlu2} L. Andersson, T. Backdahl, P. Blue and S. Ma, {\em Stability for linearized gravity on the Kerr spacetime}, arXiv:1903.03859.

\bibitem{An2020} Y. Angelopoulos, S. Aretakis and D. Gajic, {\em A Non-degenerate Scattering Theory for the Wave Equation on Extremal Reissner–Nordstr\"om}, Communications in Mathematical Physics {\bf 380}, pages 323-408 (2020).

\bibitem{Ba1991} A. Bachelot, {\em  Gravitational scattering of electromagnetic field by Schwarzschild black hole}, Annales de l'I.H.P. A {\bf 54} (1991), 261-320.

\bibitem{Ba1989a} J. C. Baez, {\em Scattering and the geometry of the solution manifold of $\square f+ \lambda f^3 =0$}, J. Funct. Anal. {\bf 83} (1989), 317-332.


\bibitem{BaSeZho1990}  J.C. Baez, I.E. Segal and Z. F. Zhou, {\em The global Goursat problem and scattering for nonlinear wave equations}, J. Funct. Anal. {\bf 93} (1990), 2, 239-269.

\bibitem{Cha1975} S. Chandrasekhar, {\it On the equations governing the perturbations of the Schwarzschild black hole}, Proc. R . Soc. Lond. A. {\bf 343}, 289-298 (1975).

\bibitem{Cha} S. Chandrasekhar, {\em The mathematical theory of black holes}, Oxford University Press 1983.

\bibitem{Da2018} M. Dafermos, I. Rodnianski, Y. Shlapentokh-Rothman, {\em A scattering theory for the wave equation on Kerr black hole exteriors}, Annales Scientifiques de l'Ecole Normale Superieure, Vol. 51, Iss.2, 371-486 (2018), arXiv:1412.8379.

\bibitem{Da2019} M. Dafermos, G. Holzegel and I. Rodnianski, {\em  The linear stability of the Schwarzschild solution to gravitational perturbations}, Acta Math., 222 (2019), 1–214.

\bibitem{Da2020} M. Dafermos, G. Holzegel and I. Rodnianski, {\em Boundedness and Decay for the Teukolsky Equation on Kerr
Spacetimes I: The Case $|a|\ll M$}, Annals of PDE, 1-118 (2019) 5:2.


\bibitem{Di1985} J. Dimock, {\em Scattering for the wave equation on the Schwarzschild metric}, Gen. Rel. Grav. {\bf 17} (1985), 4, 353--369.

\bibitem{DiKa1986} J. Dimock and B.S. Kay, {\em Classical and quantum scattering theory for linear scalar fields on the Schwarzschild metric. II}, J. Math. Phys. {\bf 27} (1986), 10, 2520--2525.

\bibitem{DiKa1987} J. Dimock and B.S. Kay, {\em Classical and Quantum Scattering theory for linear scalar fields on the Schwarzschild metric I}, Ann. Phys. {\bf 175} (1987),  366--426.


\bibitem{El2020} E. Giorgi, {\em Boundedness and decay for the Teukolsky system of spin 2 on Reissner-Nordström spacetime: the case of small charge}, 	arXiv:1811.03526.

\bibitem{El2020'} E. Giorgi, {\em The linear stability of Reissner-Nordström spacetime: the full subextremal range}, arXiv:1910.05630.

\bibitem{Fri1962} F.G. Friedlander, {\em On the radiation field of pulse solutions of the wave equation I}, Proc. Roy. Soc. Ser. A {\bf 269} (1962), 53--65.

\bibitem{Fri1964} F.G. Friedlander, {\em On the radiation field of pulse solutions of the wave equation II}, Proc. Roy. Soc. Ser. A {\bf 279} (1964), 386--394.

\bibitem{Fri1967} F.G. Friedlander, {\em On the radiation field of pulse solutions of the wave equation III}, Proc. Roy. Soc. Ser. A {\bf 299} (1967), 264--278.

\bibitem{Fri1980} F.G. Friedlander, {\em Radiation fields and hyperbolic scattering theory}, Math. Proc. Camb. Phil. Soc. {\bf 88} (1980), 483-515.

\bibitem{Fri2001} F.G. Friedlander, {\em Notes on the Wave Equation on Asymptotically Euclidean Manifolds}, J. Funct. Anal. {\bf 184} (2001), 1-18.

\bibitem{Ho1990} L. H\"ormander, {\em A remark on the characteristic Cauchy problem}, J. Funct. Anal. {\bf 93} (1990), 270--277.

\bibitem{Jo2012} J. Joudioux, {\em Conformal scattering for a nonlinear wave equation}, J. Hyperbolic Differ. Equ. {\em 9} (2012), 1, 1-65.

\bibitem{Jo2019} J. Joudioux, {\em H\"ormander's method for the characteristic Cauchy problem and conformal scattering for a non linear wave equation}, Lett. Math. Phys (2020), arXiv:1903.12591.

\bibitem{Ke2019} C. Kehle and Y. Shlapentokh-Rothman, {\em A Scattering Theory for Linear Waves on the
Interior of Reissner-Nordstr\"om Black Holes}, Ann. Henri Poincar\'e 20 (2019), 1583-1650.

\bibitem{LaPhi} P.D. Lax, R.S. Phillips, {\em Scattering theory}, Pure and Applied Mathematics, Vol. 26, Academic Press, New York-London, 1967, xii + 276 pp.(1 plate) pages.

\bibitem{Lu2010} J. Luk, {\em Improved decay for solutions to the linear wave equation on a Schwarzschild black hole}, Ann. Henri Poincar\'e {\bf 11}, 805-880 (2010).

\bibitem{MaNi2004} L.J. Mason, J.-P. Nicolas, {\em Conformal scattering and the Goursat problem}, J. Hyperbolic Differ. Equ. {\bf 1} (2004), 2, 197-233.

\bibitem{Mo2019} M. Mokdad, {\em Conformal Scattering of Maxwell fields on Reissner-Nordstr\"om-de Sitter Black Hole Spacetimes}, Annales de l'institut Fourier, {\bf 69} (2019), 5, 2291-2329.  

\bibitem{MaNi2009} L.J. Mason and J.-P. Nicolas, {\em Regularity an space-like and null infinity}, J. Inst. Math. Jussieu {\bf 8} (2009), 1, 179-208.


\bibitem{Ni1995} J.-P. Nicolas, {\em Non linear Klein-Gordon equation on Schwarzschild-like metrics}, J. Math. Pures Appl. 74 (1995), p. 35-58.

\bibitem{Ni2006} J.-P. Nicolas, {\em On Lars H\"ormander's remark on the characteristic Cauchy problem}, Annales de l'Institut Fourier, {\bf 56} (2006), 3, 517-543.

\bibitem{Ni2016} J.-P. Nicolas, {\em Conformal scattering on the Schwarzschild metric}, Annales de l'Institut Fourier, {\bf 66} (2016), 3, 1175-1216.

\bibitem{NiXu2019} J.-P. Nicolas and T.X. Pham, {\em Peeling on Kerr spacetime: linear and non linear scalar fields},  Annales Henri Poincar\'e, Vol. 20, Issue 10 (2019), p. 3419-3470.

\bibitem{Pe1964} R. Penrose, {\em Conformal approach to infinity}, in Relativity, groups and topology, Les Houches 1963, ed. B.S. De Witt and C.M. De Witt, Gordon and Breach, New-York, 1964.

\bibitem{PeRi84} R. Penrose, W. Rindler, {\em Spinors and space-time}, Vol. I \& II, Cambridge monographs on mathematical physics, Cambridge University Press 1984 \& 1986.

\bibitem{Pha2017} T.X. Pham, {\em Peeling and conformal scattering on the spacetimes of the general relativity}, Phd's thesis, Brest university (France) (4/2017), https://tel.archives-ouvertes.fr/tel-01630023/document.

\bibitem{Pha2019} T.X. Pham, {\em Peeling of Dirac field on Kerr spacetime}, Journal of Mathematical Physics 61, 032501 (2020); https://doi.org/10.1063/1.5121433.

\bibitem{ReWhe} T. Regge and John A. Wheeler, {\it Stability of a Schwarzschild singularity}, Physical Review 108 (1957), 1063-1069.

\bibitem{Ta2019} G. Taujanskas, {\em Conformal scattering of the Maxwell-scalar field system on de Sitter space}, Journal of Hyperbolic Differential Equations, Vol. 16, No. 04, pp. 743-791 (2019).


\bibitem{Teu1973} S. A. Teukolsky, {\em Perturbations of a rotating black hole. I. Fundamental equations for
gravitational, electromagnetic, and Neutrino-field perturbations}, Astrophysical J. 185
(1973), 635-648.

\bibitem{Vis} C. V. Vishveshwara, {\em Stability of the Schwarzschild metric}, Physical Review D 1 (1970), 2870-2879.

\bibitem{Ze} Frank J. Zerilli, {\em Effective potential for even-parity Regge-Wheeler gravitational perturbation equations}, Physical Review Letters 24 (1970), no. 13, 737.

\end{thebibliography}
\end{document}